\documentclass[fleqn,12pt,twoside]{article}

\usepackage[headings]{espcrc1}
\readRCS $Id: espcrc1.tex,v 1.2 2004/02/24 11:22:11 spepping Exp $
\usepackage{graphicx}
\usepackage[figuresright]{rotating}

\newtheorem{theorem}{Theorem}

\newtheorem{proposition}{Proposition}

\newenvironment{proof}[1][Proof.]{\begin{trivlist}
\item[\hskip \labelsep {\bfseries #1}]}{\end{trivlist}}

\newenvironment{acknowledgement}[1][Acknowledgement]{\begin{trivlist}
\item[\hskip \labelsep {\bfseries #1}]}{\end{trivlist}}

\newcommand{\AmS}{{\protect\the\textfont2
  A\kern-.1667em\lower.5ex\hbox{M}\kern-.125emS}}

\usepackage{amsmath}
\usepackage{amsfonts}

\begin{document}

\title{On a property of the $n$-dimensional cube}

\author{Rafayel Kamalian \address[MCSD]{Institute for Informatics and Automation Problems,\\
National Academy of Sciences of Republic of Armenia, 0014,
Armenia}
\thanks{email: rrkamalian@yahoo.com}
and
        Arpine Khachatryan
        \address{Ijevan Branch of Yerevan State University, 4001, Armenia}
\thanks{email: khachatryanarpine@gmail.com}
                        }


\runtitle{On a property of the $n$-dimensional cube}
\runauthor{Rafayel Kamalian, Arpine Khachatryan}
\maketitle

\begin{abstract}
We show that in any subset of the vertices of $n$-dimensional cube that contains at least $2^{n-1}+1$ vertices ($n\geq 4$), there are four vertices that induce a claw, or there are eight vertices that induce the cycle of length eight.
\end{abstract}

\section{Introduction and definitions}

We consider finite graphs $G=(V,E)$ with vertex set $V$ and edge set $E$. The graphs contain no multiple edges or loops. The $n$-dimensional cube will be denoted by $Q_{n}$, and claw is the complete bipartite graph $K_{1,3}$. Moreover, the vertex of degree three in the claw is called a claw-center. Non-defined terms and concepts can be found in \cite{Harary}.

The main result of the paper is the following:

\begin{theorem} Let $n\geq 4$ and let $V'\subseteq V(Q_{n})$. If $|V'|\geq 2^{n-1}+1$, then at least one of the following two conditions holds:
\begin{itemize}
	\item [(a)] there are four vertices in $V'$ that induce a claw;
	\item [(b)] there are eight vertices that induce a simple cycle.
\end{itemize}
\end{theorem}

\begin{proof} Our proof is by induction on $n$. Suppose that $n=4$. Clearly, without loss of generality, we can assume that $|V'|=9$. Consider the following partition of the vertices of $Q_{4}$:
\begin{equation*}
V_1=\{(0,\alpha_2,\alpha_3,\alpha_4):\alpha_i\in \{0,1\},2\leq i\leq 4\}, V_2=\{(1,\alpha_2,\alpha_3,\alpha_4):\alpha_i\in \{0,1\},2\leq i\leq 4\}.
\end{equation*} Clearly, the subgraphs of $Q_{4}$ induced by $V_1$ and $V_2$ are isomorphic to $Q_{3}$. Define:
\begin{equation*}
V'_1=V_1\cap V', V'_2=V_2\cap V'.
\end{equation*} We will assume that $|V'_1|\geq |V'_2|$. We will complete the proof of the base of induction, by considering the following cases:

Case 1: $|V'_1|=8$ and $|V'_2|=1$. Clearly, any vertex from $V'_1$ is a claw-center.

Case 2: $|V'_1|=7$ and $|V'_2|=2$. It is not hard to see that $V'_1$ contains a claw-center.

Case 3: $|V'_1|=6$ and $|V'_2|=3$. Again, it is a matter of direct verification that $V'_1$ contains a claw-center.

Case 4: $|V'_1|=5$ and $|V'_2|=4$. Consider the subgraph $G_1$ of $Q_4$ induced by $V'_1$. Clearly, if $G_1$ contains a vertex of degree three, then this vertex is a claw-center. Thus, any vertex in $G_1$ has degree at most two. It is not hard to see that this implies that $G_1$ contains no isolated vertex. Moreover, since $|V'_1|=5$, we have that $G_1$ is a connected graph, and therefore it is the path of length four.

Now, let $a_1,a_2,a_3$ be the internal vertices of $G_1$, and let $b_1,b_2$ be the end-vertices of $G_1$. Clearly, we can assume that neither of $a_1,a_2,a_3$ has a neighbour in $V'_2$. Since $|V_2|=8$ and $|V'_2|=4$, we have that there are five possibilities for $V'_2$. We invite the reader to check that in four of these cases one can find a claw-center in $V'_2$, and in the final case $V'$ has a vertex $z$ such that $V'\backslash \{z\}$ induces a simple cycle.\\

Now, let us assume that the statement is true for $n-1$, and let $V'\subseteq V(Q_{n})$ be a subset with $|V'|\geq 2^{n-1}+1$. Consider the following partition of the vertices of $Q_{n}$:
\begin{equation*}
V_1=\{(0,\alpha_2,...,\alpha_n):\alpha_i\in \{0,1\},2\leq i\leq n\}, V_2=\{(1,\alpha_2,...,\alpha_n):\alpha_i\in \{0,1\},2\leq i\leq n\}.
\end{equation*} Clearly, the subgraphs of $Q_{n}$ induced by $V_1$ and $V_2$ are isomorphic to $Q_{n-1}$. Moreover, it is not hard to see that at least one of the following two inequalities is true: $|V_1\cap V'|\geq 2^{n-2}+1$ and $|V_2\cap V'|\geq 2^{n-2}+1$. Thus the proof follows from the induction hypothesis. $\square$

\end{proof}

For the case of $n=3$ we have:

\begin{proposition} Let $V'\subseteq V(Q_{3})$ and let $|V'|\geq 6$. Then at least one of the following two conditions holds:
\begin{itemize}
	\item there are four vertices in $V'$ that induce a claw;
	\item there are six vertices in $V'$ that induce a simple cycle.
\end{itemize}
\end{proposition}

\begin{acknowledgement} We would like to thank Zhora Nikoghosyan and Vahan Mkrtchyan for their attention to this work.
\end{acknowledgement}

\end{document}